\documentclass[conference,onecolumn]{ieeeconf}

\IEEEoverridecommandlockouts
\usepackage[svgnames,dvipsnames]{xcolor}
\usepackage{textcomp}
\usepackage{amssymb,amsfonts,mathtools,cite,mathrsfs}
\usepackage{graphicx}
\usepackage{booktabs}
\usepackage{placeins}
\usepackage{array}
\usepackage[unicode=true,hidelinks]{hyperref}
\usepackage{booktabs}

\newcommand\scalemath[2]{\scalebox{#1}{\mbox{\ensuremath{\displaystyle #2}}}}
\newcommand{\interior}[1]{{\kern0pt#1}^{\mathrm{o}}}

\usepackage{algorithm}
\usepackage{algpseudocode}

\usepackage{pgfplots}
\pgfplotsset{compat=1.18} 

\usetikzlibrary{
    calc, decorations.pathreplacing, positioning, fit, shapes, shapes.multipart, 
    arrows, arrows.meta, shadows, backgrounds
}
\definecolor{ink}{HTML}{253345}
\definecolor{bluebg}{HTML}{F1F5FA}
\definecolor{tealbg}{HTML}{EFF7F5}
\definecolor{amber}{HTML}{B77937}
\definecolor{muted}{HTML}{647386}

\definecolor{matlabblue}{RGB}{0,114,189}
\definecolor{matlaborange}{RGB}{217,83,25}
\definecolor{matlabyellow}{RGB}{237,177,32}
\definecolor{matlabpurple}{RGB}{126,47,142}
\definecolor{matlabgreen}{RGB}{119,172,48}

\DeclareMathOperator\cbf{\scalemath{0.7}{\mathrm{cbf}}}
\DeclareMathOperator\rcbf{\scalemath{0.7}{\mathrm{rcbf}}}

\newtheorem{assumption}{Assumption}
\newtheorem{definition}{Definition}
\newtheorem{remark}{Remark}
\newtheorem{lemma}{Lemma}
\newtheorem{proposition}{Proposition}
\newtheorem{theorem}{Theorem}

\usepackage[table]{xcolor}
\definecolor{lifshade}{RGB}{237,243,248}

\definecolor{ink}{HTML}{302E2A}
\definecolor{line}{HTML}{77736B}
\definecolor{ochre}{HTML}{98763F}
\definecolor{sand}{HTML}{F3ECDE}
\definecolor{warm}{HTML}{F4F2ED}
\definecolor{edge}{HTML}{C8C1B4}

\begin{document}

\title{\LARGE \bf Compute-Constrained Safety Filters with Neuromorphic Event Triggering}

\author{Tochukwu E. Ogri$^{1}$, Opeyemi Owolabi$^{1}$, Luke Fina$^{2}$, Christopher Petersen$^{1}$ and Rushikesh Kamalapurkar$^{1}$%
\thanks{This research was supported in part by the Air Force Research Laboratory under contract number FA8651-24-1-0019. Any opinions, findings, or recommendations in this article are those of the authors and do not necessarily reflect the views of the sponsoring agencies.}%
\thanks{$^{1}$ Department of Mechanical and Aerospace Engineering, University of Florida, email: {\tt\footnotesize \{tochukwu.ogri, opeyemi.owolabi, c.petersen1, rkamalapurkar\}@ufl.edu}.}%
\thanks{$^{2}$ Naval Surface Warfare Center, Panama City Division, Panama City, FL, USA, email: {\tt \footnotesize luke.a.fina.civ@us.navy.mil.}}}

\maketitle
\thispagestyle{empty}
\pagestyle{empty}

\begin{abstract}
Safety filters that rely on control barrier function (CBF) quadratic programs (QPs) are typically analyzed under the assumption that the QP is solved instantaneously, although nonzero computation time can render a computed input unsafe by the time it is applied. This paper develops a dual leaky integrate-and-fire (LIF) event-triggered safety filter for linear systems that accounts for this delay. Two LIF states monitor a shifted CBF residual and the deviation of the held input from the nominal feedback law. A computation is initiated when either LIF state reaches its threshold; the previous input remains applied until the prescribed application time, and any solution returned earlier is stored until then. A closed-form charging profile relates the safety-LIF threshold to residual-growth bounds, ensuring sufficient residual margin over the application delay. Under the stated assumptions, the resulting hybrid system guarantees forward invariance, completeness, and non-Zeno execution, with positive lower bounds on computation-initiation and input-application times. A spacecraft terminal-approach simulation demonstrates safety with fewer optimizer calls than a delay-tightened periodic filter.
\end{abstract}

\section{Introduction}\label{sec:introduction}

    Safety-critical control requires state constraints to be enforced despite disturbances, modeling uncertainty, and computation and sampling constraints. These issues arise in applications such as spacecraft rendezvous \cite{DongSafety2017}, docking \cite{BreedenGuaranteed2022}, and proximity operations, where collision avoidance and keep-out constraints must be maintained. Existing approaches include penalty-function formulations in constrained optimal control \cite{SCC.Ogri.Qureshi.ea2026}, barrier transformations \cite{SCC.Mahmud.Nivison.ea2021}, barrier states \cite{SCC.Almubarak.Sadegh.ea2025, SCC.Satharasi.Ogri.ea2026}, and control barrier functions (CBFs) \cite{SCC.Ames.Grizzle.ea2014}.

CBFs provide conditions for enforcing state constraints through forward invariance of a prescribed safe set \cite{SCC.Wieland.Allgoewer2007, SCC.Ames.Xu.ea2017, SCC.Ames.Coogan.ea2018}. For control-affine systems, these conditions are commonly imposed through quadratic program (QP)-based safety filters that minimize deviation from a nominal input subject to the CBF and input constraints. Digital implementations use sampled measurements and zero-order-held inputs, so a barrier condition enforced at a sampling instant need not hold throughout the subsequent holding interval \cite{SCC.Singletary.Chen.ea2020}.

Sample-and-hold CBF methods address this issue by bounding inter-sample deviations. In particular, \cite{SCC.Singletary.Chen.ea2020} develops CBF conditions for sampled-data systems with known input delay, including numerical computation delay. The delay is handled within the sampled-data CBF framework, whereas the present work addresses when CBF-filter computations should be initiated under delayed input application. Separately, \cite{SCC.Bahati.Ong.ea2024} characterizes possible safety violations through input-to-state safety and develops a strengthened condition that preserves the original safe set, but assumes immediate application of the sampled input. With nonzero computation time, the previously applied input remains active while the state evolves, so a solution satisfying the CBF constraint at computation initiation need not satisfy it at application.

Repeated optimization is another implementation constraint. Explicit CBF filters \cite{SCC.Mestres.Mousavi.ea2025} reduce online computation by partitioning the state space and deriving a controller formula for each region. Event-triggered control (ETC) instead determines update times through state-dependent conditions to reduce communication or computation \cite{SCC.Tabuada2007,SCC.Yang.Belta.ea2019}. ETC has also been applied to CBF-based safety control \cite{SCC.Taylor.Ong.ea2021}. The present setting couples event-triggered computation with delayed input application: the trigger must preserve safety under the previously applied input until the prescribed application time and account for the state evolution during the delay.

In this paper, this timing requirement is addressed using leaky integrate-and-fire (LIF) dynamics inspired by neuromorphic control \cite{SCC.Medvedeva.Castanos.ea2025}. LIF dynamics have been used for event-driven plasma-density control, impulsive pendulum actuation, and spiking control of linear systems \cite{SCC.Jansen.Petri.ea2025,SCC.Petri.Postoyan.ea2025,SCC.Petri.Scheres.ea2026}. Here, LIF events schedule the computation of a control input applied after a fixed delay. The proposed dual-LIF controller uses separate safety and performance states. A shifted CBF residual drives the safety LIF, while the held-input deviation from the nominal feedback law drives the performance LIF. When either state reaches its threshold, computation is initiated; the previous input remains applied until the prescribed application time, and any solution returned earlier is retained until then.

The primary contribution of this paper is a dual-LIF event-triggering mechanism for scheduling CBF-filter computations under delayed input application. A closed-form charging profile and threshold condition determine when computation must begin so that the held input remains safe until the prescribed application time. A hybrid closed-loop model captures state sampling, computation initiation, delayed input application, and LIF reset dynamics. Under the stated assumptions, we establish strong forward invariance of an admissible hybrid-state set, completeness of maximal solutions, and positive lower bounds on consecutive jumps, computation initiations, and input applications, thereby excluding Zeno behavior. A planar spacecraft terminal-approach example compares safety margins and QP solve counts with periodic CBF implementations.

\section{Problem Formulation}\label{sec:problem}

Consider the linear system with additive disturbance
\begin{equation}
    \dot x=Ax+Bu+d(t),\qquad x(0)=x_0,
    \label{eq:plant}
\end{equation}
where $x\in\mathbb{R}^n$ is the state, $u\in\mathcal U\subset\mathbb{R}^m$ is the control input, $A\in\mathbb{R}^{n\times n}$ and $B\in\mathbb{R}^{n\times m}$ are constant matrices, and $d:\mathbb{R}_{\geq0}\to\mathbb{R}^n$ is an unknown measurable disturbance. The set $\mathcal U$ is nonempty, compact, and convex. Solutions of \eqref{eq:plant} are understood in the Carath\'eodory sense. Define the disturbance set $\mathcal W \coloneqq \{w\in\mathbb R^n:\|w\|\leq\bar d\}$. We impose the following assumption.
\begin{assumption}\label{ass:disturbance}
    There exists a known constant $\bar d\geq0$ such that the admissible disturbance satisfies $d(t) \in \mathcal{W}$ for almost every $t\geq0$.
\end{assumption}

Let $h:\mathbb{R}^n\to\mathbb{R}$ be continuously differentiable, with zero as a regular value; i.e., $\nabla h(x)\neq 0$ whenever $h(x)=0$. The associated zero superlevel set is assumed to be nonempty and compact, and is defined as
   $\mathcal S\coloneqq \left\{x\in\mathbb{R}^n:h(x)\geq0\right\}$.

The following notions of invariance and safety are adopted.
\begin{definition}\label{def:invariant}
A set $\mathcal S\subset\mathbb R^n$ is robustly forward invariant for \eqref{eq:plant} under a feedback law $\pi:\mathbb R^n\to\mathcal U$, with $u=\pi(x)$, if, for every admissible disturbance $d$ and every corresponding Carath\'eodory solution $x(\cdot)$, $x(0)\in\mathcal S$ implies $x(t)\in\mathcal S$ for all $t\geq0$.
\end{definition}

The system \eqref{eq:plant} under $u=\pi(x)$ is safe with respect to $\mathcal S$ if $\mathcal S$ is robustly forward invariant.
    
Forward invariance is characterized by the behavior of the closed-loop vector field on the boundary of the safe set. By Nagumo's theorem \cite{SCC.BLANCHINI.ea1999}, for a locally Lipschitz continuous drift $F$, the set $\mathcal S$ is forward invariant for the closed-loop system $\dot x=F(x)$ if and only if
$h(x)=0$ implies $\nabla h(x)F(x)\geq0$.
CBFs provide a constructive framework to enforce this boundary condition by requiring a state-dependent inequality whose restriction to $h(x)=0$ reduces to $\nabla h(x)F(x)\geq0$. For the disturbance model in \eqref{eq:plant}, we use the following robust CBF definition.
\begin{definition}\label{def:cbf}
    A continuously differentiable function $h:\mathbb{R}^n\to\mathbb{R}$ is a CBF for \eqref{eq:plant} if there exists an extended class-$\mathcal K$ function $\alpha$ such that for all $x\in\mathcal S$,
    \begin{equation}
        \sup_{u\in\mathcal U} \left[ \nabla h(x)(Ax+Bu) -\bar d\|\nabla h(x)\| \right] \geq -\alpha\bigl(h(x)\bigr).
        \label{eq:cbf_condition}
    \end{equation}
\end{definition}

To ensure the well-posedness of the safety filters, we impose the following regularity conditions.

\begin{assumption}\label{ass:cbf_regularity}
There exists an open neighborhood of $\mathcal S$ on which $\nabla h$ is locally Lipschitz continuous. The extended class-$\mathcal K$ function $\alpha$ is locally Lipschitz continuous on an open interval containing $\{h(x):x\in\mathcal S\}$.
\end{assumption}

Define the disturbance-adjusted residual $r(x,u) \coloneqq -\nabla h(x)(Ax+Bu) - \alpha(h(x)) + \bar d\|\nabla h(x)\|$. For each fixed $x$, the mapping $u\mapsto r(x,u)$ is affine. Under Assumption~\ref{ass:cbf_regularity}, there exists an open neighborhood of $\mathcal S$ on which the mapping $x\mapsto r(x,u)$ is locally Lipschitz continuous, uniformly over $u\in\mathcal U$. The corresponding set of admissible inputs is given by
\begin{equation}
    \mathcal{K}_{\cbf}(x) \coloneqq \{u \in \mathcal{U} : r(x,u) \leq 0\}.
    \label{eq:cbf_admissible_set}
\end{equation}
The following result establishes robust forward invariance under any locally Lipschitz continuous CBF-admissible feedback policy.
\begin{theorem}\label{thm:cbf_invariance}
Suppose Assumption~\ref{ass:disturbance} holds and $h$ satisfies \eqref{eq:cbf_condition}. If a locally Lipschitz continuous policy $\pi:\mathbb{R}^n\to\mathcal U$ satisfies
\begin{equation}
    \pi(x)\in\mathcal K_{\cbf}(x), \qquad \forall x\in\mathcal S, \label{eq:cbf_feedback_condition}
\end{equation}
then $\mathcal S$ is robustly forward invariant for $\dot x=Ax+B\pi(x)+d(t)$.
\end{theorem}
\begin{proof}
Define $\mathcal F_\pi(x)\coloneqq\{Ax+B\pi(x)+w:w\in\mathcal W\}$. For every $x\in\partial\mathcal S$, \eqref{eq:cbf_feedback_condition}, $h(x)=0$, and $\alpha(0)=0$ imply $\min_{w\in\mathcal W}\nabla h(x)(Ax+B\pi(x)+w)=-r(x,\pi(x))\geq0$. Since zero is a regular value of $h$, the Clarke tangent cone of $\mathcal S$ at $x$ satisfies $T_{\mathcal S}(x)=\{v\in\mathbb R^n:\nabla h(x) v\geq0\}$ for $x\in\partial\mathcal S$, while $T_{\mathcal S}(x)=\mathbb R^n$ for $x\in\operatorname{int}\mathcal S$. Hence $\mathcal F_\pi(x)\subseteq T_{\mathcal S}(x)$ for all $x\in\mathcal S$.

Because $\pi$ is locally Lipschitz continuous, $x\mapsto Ax+B\pi(x)$ is locally Lipschitz continuous. Since $\mathcal W$ is nonempty, compact, and convex, $\mathcal F_\pi$ has nonempty compact convex values and is locally bounded. Moreover, translation by the fixed set $\mathcal W$ preserves the local Lipschitz bound of $x\mapsto Ax+B\pi(x)$, so $\mathcal F_\pi$ is locally Lipschitz continuous as a set-valued map. Thus, $\mathcal F_\pi$ satisfies the standing hypotheses and the local-Lipschitz condition of \cite[Theorem~1]{SCC.Usevitch.Garg.ea2020}. Since $\mathcal F_\pi(x)\subseteq T_{\mathcal S}(x)$ for all $x\in\mathcal S$, \cite[Theorem~1]{SCC.Usevitch.Garg.ea2020} implies strong invariance of $\mathcal S$ for $\dot x\in\mathcal F_\pi(x)$.

Since every Carath\'eodory solution of \eqref{eq:plant} generated by an admissible disturbance $d$ is a solution of this inclusion. Thus, $\mathcal S$ is robustly forward invariant for \eqref{eq:plant} under $u=\pi(x)$.
\end{proof}

Let $u_{\mathrm{des}}:\mathcal S\to\mathcal U$ be a locally Lipschitz continuous nominal feedback law. Define the objective function
    $\varphi(x, u) \coloneqq \frac{1}{2}\|u-u_{\mathrm{des}}(x)\|^2$.
The pointwise CBF filter to ensure safety is given by (cf. \cite{SCC.Ames.Coogan.ea2018})
\begin{align}
    \pi_0(x) \coloneqq \underset{u\in\mathcal{U}}{\operatorname{arg\,min}} \quad & \varphi(x, u) \nonumber \\
    \text{subject to} \quad & r(x, u) \leq 0.
    \label{eq:standard_cbf_filter}
\end{align}
When $\mathcal U$ is polyhedral, \eqref{eq:standard_cbf_filter} is a quadratic program; for general compact convex $\mathcal U$, it is a convex quadratic optimization problem. Provided \eqref{eq:standard_cbf_filter} is feasible and the minimizer $\pi_0$ is locally Lipschitz continuous, Theorem~\ref{thm:cbf_invariance} guarantees that the continuous feedback law $u=\pi_0(x)$ renders $\mathcal S$ robustly forward invariant.

The optimization problem in \eqref{eq:standard_cbf_filter} is generally assumed to be solved instantaneously \cite{SCC.Ames.Coogan.ea2018}. In this paper, feasible, but possibly suboptimal, control inputs are assumed to become available within at most $T$ seconds of computation initiation and are applied at a prescribed time $T$ seconds after initiation. This nonzero computation time introduces two challenges. First, if a computation is initiated at $t_k$ using the sampled state $y=x(t_k)$, the resulting input $u_k$ is applied at $t_k+T$, after the state has evolved. Hence, although $u_k$ satisfies the tightened CBF constraint at $x(t_k)$, it must also satisfy the CBF condition \eqref{eq:cbf_condition} at the application state $x(t_k+T)$. Second, during $[t_k,t_k+T)$, the previously applied input $u_{k-1}$ remains active and must preserve safety until $u_k$ is applied. After application, $u_k$ remains constant until the next input $u_{k+1}$ is applied at $t_{k+1}+T$. We address the first challenge through constraint tightening and the second by initiating the next computation sufficiently early, based on the available residual margin. The objective of this paper is to design the tightened constraints and the triggers that request new solutions of \eqref{eq:standard_cbf_filter}.

\section{Constraint Tightening}
To account for the evolution of the system state between computation initiation and input application, the CBF constraint needs to be tightened by a margin $\delta>0$. Define the largest uniform feasible tightening as
\begin{equation}
    \delta_{\mathrm{feas}} \coloneqq \min_{x\in\mathcal S} \max_{v\in\mathcal U}\{-r(x,v)\}.
    \label{eq:delta_feas}
\end{equation}
For each $x\in\mathcal S$, the inner maximum is attained because $\mathcal U$ is compact and $r$ is continuous. Moreover, the map $x\mapsto\max_{v\in\mathcal U}\{-r(x,v)\}$ is continuous; hence, since $\mathcal S$ is compact, the minimum in \eqref{eq:delta_feas} is attained.
\begin{assumption}\label{ass:filter}
The uniform feasibility margin satisfies $\delta_{\mathrm{feas}}>0$.
\end{assumption}
For any $0<\delta<\delta_{\mathrm{feas}}$, define the $\delta$-tightened admissible control set
\begin{equation}
    \mathcal{K}_{\rcbf}^{\delta}(x) \coloneqq \left\{u\in\mathcal{U} : r(x,u)\leq -\delta\right\}. \label{eq:robust_set}
\end{equation}
By the definition of $\delta_{\mathrm{feas}}$ in \eqref{eq:delta_feas}, $\mathcal K_{\rcbf}^{\delta}(x)$ is nonempty for every $x\in\mathcal S$. 

To model early termination with a feasible suboptimal input, define the optimal value function $V^{*}(x)\coloneqq \min_{v\in\mathcal K_{\rcbf}^{\delta}(x)}\varphi(x,v)$ and, for a solver tolerance $\epsilon\geq0$, the $\epsilon$-suboptimal solution map as
\begin{equation}
    \Psi_{\epsilon}(x) \coloneqq \left\{ v\in\mathcal{K}_{\rcbf}^{\delta}(x) : \varphi(x, v)\leq V^{*}(x)+\epsilon \right\}. \label{eq:solution_map}
\end{equation}
That is, $\epsilon$ affects objective suboptimality but not feasibility of the tightened constraint.

\begin{proposition}\label{prop:filter_well_posed}
Let $0<\delta<\delta_{\mathrm{feas}}$. If Assumptions \ref{ass:disturbance}--\ref{ass:filter} hold, then the set-valued map $\mathcal K_{\rcbf}^{\delta}:\mathcal S\rightrightarrows\mathcal U$ is continuous on $\mathcal S$ and has nonempty compact convex values. Moreover, $V^*$ is continuous, $\Psi_\epsilon$ has nonempty compact values and is outer semicontinuous on $\mathcal S$, and the minimizer $\pi_\delta(x)\coloneqq\operatorname*{arg\,min}_{v\in\mathcal K_{\rcbf}^{\delta}(x)}\varphi(x,v)$ is unique and continuous. For every $v\in\Psi_\epsilon(x)$,
    $\|v-\pi_\delta(x)\|\leq\sqrt{2\epsilon}$.
\end{proposition}
\begin{proof}
    The proof is given in Appendix~\ref{app:prop_filter_proof}.
\end{proof}

Note that if $v_k\in\Psi_\epsilon(y)$ is applied at $t_k+T$ and $u_{\mathrm{des}}$ has Lipschitz constant $L_{\mathrm{des}}$ on $\mathcal S$, then
\begin{equation}
    \|v_k-u_{\mathrm{des}}(x(t_k+T))\| \leq \sqrt{2(V^*(y)+\epsilon)}+L_{\mathrm{des}}MT.
    \label{eq:application_performance_bound}
\end{equation}

The following section details the trigger condition used to ensure that a solution to \eqref{eq:standard_cbf_filter} is requested far enough away from the boundary of $\mathcal S$ to ensure safety in spite of a piecewise constant control signal.


\section{Compute-Constrained Event-Triggered CBF Architecture}
\label{sec:zoh_formulation}

Let $t_k$ denote the time instant when the $k$-th solution of \eqref{eq:standard_cbf_filter} is requested and $y \coloneqq x(t_k)$. Let $u_k \in \Psi_\epsilon (y)$ be a feasible solution returned by the solver for the problem initiated at $t_k$. The solution takes at most $T$ seconds to compute, i.e., it becomes available at some time $\hat t_k \in (t_k, t_k+T]$.

In this paper, we assume that the control input $u_{k-1}$, i.e., a feasible solution of \eqref{eq:standard_cbf_filter} initiated at time $t_{k-1}$ is applied to the system at time $t_{k-1}+T$ and kept constant over the interval $[t_{k-1}+T, t_k+T)$. At time $t_k + T$, the control input $u_k$, i.e., a feasible solution of \eqref{eq:standard_cbf_filter} initiated at time $t_k$ is applied and kept constant over $[t_k + T, t_{k+1} + T)$, resulting in a piecewise constant control signal. The timing convention is illustrated in Fig.~\ref{fig:timing}.
\begin{figure}[t]
    \centering
    \resizebox{\columnwidth}{!}{

\begin{tikzpicture}[
    font=\sffamily\small,
    text=ink,
    axis/.style={
        draw=line,
        line width=.7pt,
        -{Stealth[length=1.8mm]}
    },
    tick/.style={
        draw=line,
        line width=.7pt
    },
    hold/.style={
        draw=edge,
        blue,
        line width=2.2pt
    },
    comp/.style={
        draw=ochre,
        line width=2.2pt
    },
    store/.style={
        draw=ochre,
        line width=2.2pt,
        dash pattern=on 3pt off 2pt
    },
    lab/.style={
        font=\sffamily\footnotesize,
        text=ink,
        align=center
    },
    olab/.style={
        font=\sffamily\footnotesize,
        text=ochre,
        align=center
    },
    blab/.style={
        font=\sffamily\footnotesize,
        text=blue,
        align=center
    }
]

\def\tprev{0}
\def\tk{3.1}
\def\that{5.5}
\def\tapp{7.2}
\def\tnext{9.5}

\draw[axis] (-0.4,0) -- (\tnext+0.7,0)
    node[right] {$t$};

\foreach \x in {\tprev,\tk,\that,\tapp,\tnext}
    \draw[tick] (\x,-0.10) -- (\x,0.10);

\node[lab,below=4pt] at (\tprev,0) {$t_{k-1}+T$};
\node[lab,below=4pt] at (\tk,0) {$t_k$};
\node[lab,below=4pt] at (\that,0) {$\hat t_k$};
\node[lab,below=4pt] at (\tapp,0) {$t_k+T$};
\node[lab,below=4pt] at (\tnext,0) {$t_{k+1}$};

\draw[hold] (\tprev,1.45) -- (\tapp,1.45);
\draw[axis,blue] (\tprev,2.25) -- (\tprev,1.62);
\node[blab,above=1pt] at (\tprev,2.25)
    {apply $u_{k-1}$};
\node[lab,blue,xshift=15mm] at (\tprev,1.7)
    {$u = u_{k-1}$};

\draw[hold,magenta] (\tapp,1.45) -- ({\tnext+0.75},1.45);
\node[lab,magenta,xshift=15mm] at (\tapp,1.7)
    {$u = u_{k}$};
    
\draw[comp] (\tk,0.85) -- (\that,0.85);
\node[olab] at ({(\tk+\that)/2},1.08)
    {solver active};

\draw[store] (\that,0.85) -- (\tapp,0.85);
\node[olab] at ({(\that+\tapp)/2},1.08)
    {$u_k$ available};

\draw[axis] (\tk,2.25) -- (\tk,1.62);
\node[lab,above=1pt] at (\tk,2.25)
    {new solution request};

\draw[axis,magenta] (\tapp,2.25) -- (\tapp,1.62);
\node[olab,magenta,above=1pt] at (\tapp,2.25)
    {apply $u_k$};

\draw[axis] (\tnext,2.25) -- (\tnext,1.62);
\node[lab,above=1pt] at (\tnext,2.25)
    {new solution request};

\draw[decorate,decoration={brace,amplitude=4pt}]
    (\tprev,2.75) -- (\tk,2.75)
    node[midway,above=4pt,lab] {idle mode};

\draw[decorate,decoration={brace,amplitude=4pt}]
    (\tk,2.75) -- (\tapp,2.75)
    node[midway,above=4pt,lab] {application-pending mode};


\end{tikzpicture}
}\vspace{-1em}
    \caption{Timing of one update cycle. The input $u_{k-1}$ remains applied from $t_{k-1}+T$ until $t_k+T$. A computation is initiated at $t_k$. A solution $u_k$ is returned at some $\hat t_k\in(t_k,t_k+T]$ and applied at $t_k+T$.}
    \label{fig:timing}
\end{figure}
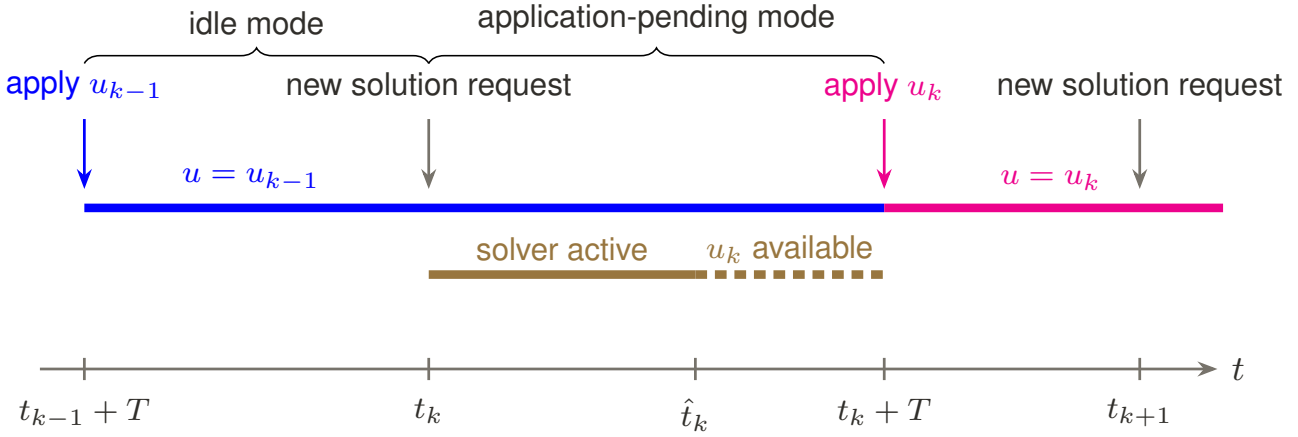

\subsection{Residual Growth Bounds}

Safety while $u_{k-1}$ is applied is quantified by the residual $r(x(t),u_{k-1})$. During $[t_k,t_k+T)$, the candidate input $u_k$ is not yet applied, but the variation of $r(x(t),u_k)$ determines whether $u_k$ satisfies the CBF condition \eqref{eq:cbf_condition} at the application time $t_k+T$. We therefore bound the variation of both the residual associated with the held input and the candidate residual associated with the input to be applied at $t_k+T$. Since $\mathcal S$, $\mathcal U$, and $\mathcal W$ are compact, define $M\coloneqq\max_{x\in\mathcal S,u\in\mathcal U,w\in\mathcal W}\|Ax+Bu+w\|$. Let $L_r>0$ be a uniform Lipschitz constant for $x \mapsto r(x, u)$ on $\mathcal S\times\mathcal U$. Select uniform constants $a,b>0$ such that, along every admissible constant-input trajectory contained in $\mathcal S$,
\begin{align}
    \frac{d}{dt}r(x(t),u) &\leq a \quad \text{a.e.}, \label{eq:residual_drift_a}\\
    \frac{d}{dt}r(x(t),v) &\leq b \quad \text{a.e.},\qquad \forall v\in\mathcal U.\label{eq:residual_drift_b}
\end{align}
On any interval $[s,t]$ over which $u$ is constant and $x([s,t])\subset\mathcal S$, every admissible disturbance satisfies $\|\dot x(\tau)\|\leq M$ almost everywhere on $[s,t]$. Therefore, one admissible choice is $a=b=L_rM$.

If a computation is initiated at time $t_k$ with the applied input $u_{k-1}$ satisfying $r(x(t_k),u_{k-1})\leq-aT$, then, for every $t\in[t_k,t_k+T]$ such that $x([t_k,t])\subset\mathcal S$,
\begin{multline}
    r(x(t),u_{k-1}) \leq r(x(t_k),u_{k-1})+a(t-t_k) \\
    \leq -aT+a(t-t_k) \leq 0.
    \label{eq:held_input_safety}
\end{multline}

For $y=x(t_k)$ and $u_k\in\Psi_\epsilon(y)$, $r(y,u_k)\leq-\delta$. If $x([t_k,t_k+T])\subset\mathcal S$, then integrating \eqref{eq:residual_drift_b} over $[t_k,t_k+T]$ gives $r(x(t_k+T),u_k)\leq-\delta+bT$. Therefore, the condition
\begin{equation}
    (a+b)T+\kappa<\delta<\delta_{\mathrm{feas}} \label{eq:zoh_margin_condition}
\end{equation}
implies $r(x(t_k+T),u_k)<-aT-\kappa$. The interval in \eqref{eq:zoh_margin_condition} is nonempty if and only if $(a+b)T+\kappa<\delta_{\mathrm{feas}}$.

\subsection{Dynamic Event-Triggering Mechanism}
\label{subsec:dynamic_triggers}

The control input is only updated on an event-triggered basis. If events were triggered only for safety based on the safety residual, then the applied control can deviate significantly from the desired nominal control, resulting in loss of performance. To limit such performance degradation, update events are also triggered based on the control deviation
\begin{equation}
    e(x,u)\coloneqq\|u-u_{\mathrm{des}}(x)\|.
    \label{eq:zoh_performance_error}
\end{equation}
The residual $r(x,u)$ and deviation $e(x,u)$ define the safety and performance LIF driving signals, respectively, defined as
\begin{align}
    \psi_{s}(x,u) &\coloneqq \max\left\{0,\, r(x,u)+aT+\kappa\right\}, \text{ and }\label{eq:zoh_safety_drive}\\
    \psi_{p}(x,u) &\coloneqq \max\left\{0,\, e(x,u)-u_{\mathrm{tol}}\right\}, \label{eq:zoh_performance_drive}
\end{align}
where $\kappa>0$ and $u_{\mathrm{tol}}\geq0$ are user-selected design parameters; $\kappa$ determines the width of the residual interval $[-aT-\kappa,-aT]$, while $u_{\mathrm{tol}}$ specifies the performance-deviation threshold. Note that, $\psi_s(x,u)>0$ if and only if $r(x,u)>-aT-\kappa$, and $\psi_p(x,u)>0$ if and only if $e(x,u)>u_{\mathrm{tol}}$.

In the following, we define the event-triggering mechanism using two LIF states $\xi_i$, $i\in\{s,p\}$, driven by the corresponding signals $\psi_i$. While no computation is active, the safety and performance LIF states $\xi_s,\xi_p\in\mathbb R_{\geq0}$ satisfy
\begin{align}
    \dot{\xi}_{s} &= -\mu_{s}\xi_{s} + k_{s}\psi_{s}(x,u), \label{eq:zoh_safety_filter}\\
    \dot{\xi}_{p} &= -\mu_{p}\xi_{p} + k_{p}\psi_{p}(x,u), \label{eq:zoh_performance_filter}
\end{align}
where $\mu_\ell\geq0$ and $k_\ell>0$ are the leakage rate and input gain, respectively, for $\ell\in\{s,p\}$. Thus, $\xi_s$ accumulates the safety driving signal $\psi_s$ subject to leakage $\mu_s\xi_s$, and $\xi_p$ accumulates the performance driving signal $\psi_p$ subject to leakage $\mu_p\xi_p$. A computation is initiated when $\xi_s=\Delta_s$ or $\xi_p=\Delta_p$, where $\Delta_s,\Delta_p>0$ are the corresponding firing thresholds. The LIF states are frozen during the application-pending mode, corresponding to $[t_k,t_k+T)$, and reset to zero when $u_k$ is applied at $t_k+T$.
The safety threshold $\Delta_s$ is selected to ensure that a computation is initiated no later than the time at which the residual reaches $-aT$. Since $\psi_s(x,u)>0$ for $r(x,u)>-aT-\kappa$, define
\begin{equation}
    \Phi_{s}(\rho) \coloneqq
    \begin{cases}
        \dfrac{k_{s}\rho^2}{2a}, & \mu_{s}=0,\\[10pt]
        \dfrac{k_sa}{\mu_s^2}\left(\dfrac{\mu_{s}\rho}{a}-1 +e^{-\mu_{s}\rho/a}\right), & \mu_{s}>0,
    \end{cases}
    \label{eq:zoh_charging_profile}
\end{equation}
for $\rho\geq0$. For $\rho>0$, $\Phi_s'(\rho)=k_s\rho/a>0$ when $\mu_s=0$, while $\Phi_s'(\rho)=\frac{k_s}{\mu_s}(1-e^{-\mu_s\rho/a})>0$ when $\mu_s>0$. Hence $\Phi_s$ is strictly increasing on $(0,\infty)$. Select
\begin{equation}
    0<\Delta_s\leq\Phi_s(\kappa).\label{eq:zoh_firing_condition}
\end{equation}
The following lemma shows that \eqref{eq:zoh_firing_condition} forces the safety LIF state to fire no later than the time at which the residual reaches $-aT$.
\begin{lemma}\label{lem:safety_filter_trigger}
Suppose \eqref{eq:zoh_firing_condition} holds, $u\in\mathcal U$ is constant on $[t_a,t_b]$, $x(t)\in\mathcal S$ for all $t\in[t_a,t_b]$, and $r(x(t_a),u)=-aT-\kappa$ and $r(x(t_b),u)=-aT$. If $\xi_s(t_a)\in[0,\Delta_s]$ and $\xi_s$ evolves according to \eqref{eq:zoh_safety_filter} until it reaches $\Delta_s$, then there exists $t^*\in[t_a,t_b]$ such that $\xi_s(t^*)=\Delta_s$.
\end{lemma}
\begin{proof}
 The proof is given in Appendix~\ref{app:lemma_filter_proof}.
\end{proof}

The closed-loop system, including the trigger conditions and delayed input-update mechanism, is formalized next as a hybrid system.


\subsection{Hybrid Dynamical System}

Define the hybrid state $z = (x,u,y,\xi_{s},\xi_{p},\tau,q) \in \mathcal Z$, on $\mathcal Z \coloneqq \mathcal{S} \times \mathcal U \times \mathcal{S} \times [0,\Delta_{s}] \times [0,\Delta_{p}] \times [0,T] \times \{0,1\}$. Here, $y$ stores the state sampled at the initiation of a computation, $\tau$ is the time remaining before the next input application, and $q=0$ and $q=1$ denote the idle and application-pending modes, respectively. 

Define the mode-dependent flow sets $\mathcal C_i \coloneqq \{z\in\mathcal Z : q=0,\ \tau=0\}$ and $\mathcal C_a \coloneqq \{z\in\mathcal Z : q=1\}$, and let $\mathcal C \coloneqq \mathcal C_i \cup \mathcal C_a$. The set-valued flow map is defined as
\begin{equation}
    F(z)\coloneqq \left\{\begin{bmatrix}
        Ax+Bu+w\\
        0_m\\
        0_n\\
        (1-q)(-\mu_s\xi_s+k_s\psi_s(x,u))\\
        (1-q)(-\mu_p\xi_p+k_p\psi_p(x,u))\\
        -q\\
        0
    \end{bmatrix}:w\in\mathcal W\right\}.
    \label{eq:zoh_flow_map}
\end{equation}

In the following,  jumps due to the LIF states reaching their firing thresholds are referred to as initiation jumps and those due to application of an updated control input (i.e., $\tau$ reaching zero in the application-pending mode) are referred to as application jumps. Define
$\mathcal D_i\coloneqq\{z\in\mathcal C_i:\xi_s=\Delta_s\text{ or }\xi_p=\Delta_p\}$, $\mathcal D_a\coloneqq\{z\in\mathcal C_a:\tau=0\}$, and
$\mathcal D\coloneqq\mathcal D_i\cup\mathcal D_a$. The jump maps are defined as
\begin{align}
    G_i(z)&\coloneqq \{(x,u,x,\xi_s,\xi_p,T,1)\}, && z\in\mathcal D_i, \label{eq:zoh_start_map}\\
    G_a(z)&\coloneqq \{(x,v,y,0,0,0,0):v\in\Psi_\epsilon(y)\}, && z\in\mathcal D_a.
    \label{eq:zoh_finish_map}
\end{align}
The closed-loop system is represented by the hybrid inclusion $\mathcal H\coloneqq(\mathcal C,F,\mathcal D,G)$, where $G=G_i$ on $\mathcal D_i$ and $G=G_a$ on $\mathcal D_a$. At points in $\mathcal C\cap\mathcal D$, $\mathcal H$ permits either flow or jump. Algorithm~\ref{alg:zoh_execution} selects a jump whenever $z\in\mathcal D$.

\begin{algorithm}[h!]
\small
\caption{Event-Triggered Dual-LIF Controller}
\label{alg:zoh_execution}
\begin{algorithmic}[1]
\Require $x_0\in\mathcal S$, $u_0\in\mathcal K_{\rcbf}^{\delta}(x_0)$; parameters satisfying \eqref{eq:zoh_margin_condition} and \eqref{eq:zoh_firing_condition}
\State \textbf{Initialize} $z\gets(x_0,u_0,x_0,0,0,0,0)$
\Loop
    \State \textbf{Flow:} $\dot z\in F(z)$ \textbf{until} $z\in\mathcal D_i$
    \Statex \hspace{\algorithmicindent}\Comment{Idle mode: evolve the LIF states}
    \State \textbf{Jump (Initiation):} $(y,\tau,q)\gets(x,T,1)$
    \State \textbf{Trigger} asynchronous solver for $u_k\in\Psi_\epsilon(y)$
    \Statex \hspace{\algorithmicindent}\Comment{Retain $u_k$ until $\tau=0$ if returned early}
    \State \textbf{Flow:} $\dot z\in F(z)$ \textbf{until} $z\in\mathcal D_a$
    \Statex \hspace{\algorithmicindent}\Comment{Application-pending mode: decrease $\tau$}
    \State \textbf{Jump (Application):} $(u,\xi_s,\xi_p,\tau,q)\gets(u_k,0,0,0,0)$
\EndLoop
\end{algorithmic}
\end{algorithm}

\section{Safety and Execution Guarantees}
\label{sec:zoh_analysis}

The following proposition verifies the hybrid basic conditions for $\mathcal H$.
\begin{proposition}\label{prop:zoh_basic_conditions}
Under Assumptions~\ref{ass:disturbance} and \ref{ass:cbf_regularity}, and \eqref{eq:zoh_margin_condition}, $\mathcal H$ satisfies the hybrid basic conditions of \cite{SCC.Goebel.Sanfelice.ea2012}.
\end{proposition}
\begin{proof}
The sets $\mathcal C$ and $\mathcal D$ are closed subsets of the compact set $\mathcal Z$. Continuity of the flow vector in $(z,w)$ and compactness and convexity of $\mathcal W$ imply that $F$ is outer semicontinuous and locally bounded on $\mathcal C$, with nonempty convex values. Proposition~\ref{prop:filter_well_posed} gives nonempty compact values and outer semicontinuity of $\Psi_\epsilon$. Hence, $G_i$ and $G_a$ are outer semicontinuous and locally bounded on their respective domains, with nonempty values. Since $q\in\{0,1\}$, any convergent sequence in $\mathcal D$ has an eventually constant mode component. Hence, such a sequence eventually lies entirely in either $\mathcal D_i$ or $\mathcal D_a$, and the outer semicontinuity of $G_i$ and $G_a$ implies that $G$ is outer semicontinuous on $\mathcal D$.
\end{proof}

To analyze invariance properties of $\mathcal H$, define
    $\Gamma(x,y) \coloneqq \max_{v\in\mathcal{U}} \{r(x,v)-r(y,v)\}$.
Continuity of $r$ and compactness of $\mathcal U$ imply that $\Gamma$ is continuous. Define the sets 
\begingroup\medmuskip=0mu\begin{align}
    \mathcal E_i &\coloneqq
    \{z\in\mathcal C_i:\xi_s\geq\Phi_s(\psi_s(x,u))\}, \label{eq:zoh_idle_tube}\\
    \mathcal E_a &\coloneqq \scalemath{0.97}{\{z\in\mathcal C_a:r(x,u)\leq-a\tau,\ \Gamma(x,y)\leq b(T-\tau)\}},
    \label{eq:zoh_compute_tube}
\end{align}\endgroup
and $\mathcal E\coloneqq\mathcal E_i\cup\mathcal E_a$, where $a$ and $b$ are introduced in \eqref{eq:residual_drift_a} and \eqref{eq:residual_drift_b}. The sets $\mathcal E_i$ and $\mathcal E_a$ are compact. For $z\in\mathcal E_i$, $\Phi_s(\psi_s(x,u))\leq\xi_s\leq\Delta_s\leq\Phi_s(\kappa)$; since $\Phi_s$ is strictly increasing, $\psi_s(x,u)\leq\kappa$, and therefore $r(x,u)\leq-aT<0$. For $z\in\mathcal E_a$, $r(x,u)\leq-a\tau\leq0$.

The next result establishes invariance of $\mathcal E$ under flows and jumps.

\begin{proposition}\label{prop:zoh_preinvariance}
Under \eqref{eq:zoh_margin_condition} and \eqref{eq:zoh_firing_condition}, every solution $\phi$ of $\mathcal H$ with $\phi(0,0)\in\mathcal E$ satisfies $\phi(t,j)\in\mathcal E$ for all $(t,j)\in\operatorname{dom}\phi$.
\end{proposition}
\begin{proof}
Consider an idle flow on a hybrid interval $I_j$ for a fixed $j$ beginning at $(t_0,j)$ with $\phi(t_0,j)\in\mathcal E_i$ and define $\sigma(t)\coloneqq\psi_s(x(t,j),u(t,j))$ and $\nu(t)\coloneqq\xi_s(t,j)-\Phi_s(\sigma(t))$. Both functions are locally absolutely continuous on $I_j$. Since $u$ is constant during idle flow, \eqref{eq:residual_drift_a} and the fact that $s\mapsto\max\{0,s\}$ is nondecreasing and $1$-Lipschitz imply that $\dot{\sigma}(t)\leq a$ almost everywhere. Moreover, $\Phi_s'(\rho)\geq0$ and $a\Phi_s'(\rho)=-\mu_s\Phi_s(\rho)+k_s\rho$ for $\rho\geq0$. Hence $\dot{\nu}\geq-\mu_s\nu$ almost everywhere on $I_j$, and therefore $\nu(t)\geq e^{-\mu_s(t-t_0)}\nu(t_0)\geq0$ for all $t\geq t_0$ in $I_j$. Thus $\phi(t,j)\in\mathcal E_i$ throughout the idle flow.

Now consider an application-pending flow beginning at $(t_0,j)$ with $\phi(t_0,j)\in\mathcal E_a$. For $s\geq0$ such that $(t_0+s,j)\in\operatorname{dom}\phi$, $\tau(t_0+s,j)=\tau(t_0,j)-s$. Integrating \eqref{eq:residual_drift_a} over $[t_0,t_0+s]$ gives $r(x(t_0+s,j),u(t_0+s,j))\leq r(x(t_0,j),u(t_0,j))+as\leq-a\tau(t_0+s,j)$. Moreover, for every $v\in\mathcal U$, integrating \eqref{eq:residual_drift_b} over $[t_0,t_0+s]$ and using that $y$ is constant during the application-pending flow gives $r(x(t_0+s,j),v)-r(y(t_0+s,j),v)\leq\Gamma(x(t_0,j),y(t_0,j))+bs\leq b(T-\tau(t_0+s,j))$. Taking the maximum over $v\in\mathcal U$ yields the second defining inequality of $\mathcal E_a$.

At an initiation jump from $\mathcal E_i$, $r(x,u)\leq-aT$, $y^+=x$, and $\tau^+=T$. Since $\Gamma(x,x)=0$, every jump image belongs to $\mathcal E_a$. At an application jump from $\mathcal E_a$, $\tau=0$ and every $v\in\Psi_\epsilon(y)$ satisfies
$r(x,v)\leq r(y,v)+\Gamma(x,y)\leq-\delta+bT<-aT-\kappa$.
Hence $\psi_s(x,v)=0$ and $\xi_s^+=0=\Phi_s(0)$, so every jump image belongs to $\mathcal E_i$.

Since $\phi(0,0)\in\mathcal E=\mathcal E_i\cup\mathcal E_a$, the corresponding flow argument applies to the initial hybrid interval. Each initiation jump maps $\mathcal E_i$ into $\mathcal E_a$, and each application jump maps $\mathcal E_a$ into $\mathcal E_i$. Repeating the corresponding flow and jump arguments over successive hybrid intervals yields $\phi(t,j)\in\mathcal E$ for all $(t,j)\in\operatorname{dom}\phi$.
\end{proof}

The following lemma establishes local existence from every point in $\mathcal E$.

\begin{lemma}\label{lem:zoh_existence}
For every $z_0\in\mathcal E$, there exists a nontrivial solution $\phi$ of $\mathcal H$ with $\phi(0,0)=z_0$, i.e., $\operatorname{dom}\phi\neq\{(0,0)\}$.
\end{lemma}
\begin{proof}
If $z_0\in\mathcal D$, then $G(z_0)$ is nonempty, and \cite[Prop.~2.10]{SCC.Goebel.Sanfelice.ea2012} guarantees a nontrivial solution of $\mathcal H$ from $z_0$.

Suppose $z_0\notin\mathcal D$. In the idle mode, $r(x_0,u_0)\leq-aT<0$ and both LIF states are below their firing thresholds. In the application-pending mode, $\tau_0>0$ and $r(x_0,u_0)\leq-a\tau_0<0$. For the fixed input $u_0$, define $\mathcal F_{u_0}(x)\coloneqq\{Ax+Bu_0+w:w\in\mathcal W\}$. The map $\mathcal F_{u_0}$ has nonempty compact convex values and is locally Lipschitz continuous as a set-valued map. Hence, by a standard existence result for differential inclusions \cite[Chp.~2]{SCC.Aubin.Cellina1984}, the inclusion $\dot x\in\mathcal F_{u_0}(x)$ admits an absolutely continuous local solution from $x_0$.

If $x_0\in\operatorname{int}\mathcal S$, continuity of the local solution of $\dot x\in\mathcal F_{u_0}(x)$ implies $x(t)\in\mathcal S$ for all sufficiently small $t>0$. If $x_0\in\partial\mathcal S$, then $h(x_0)=0$ and 
    $\min_{w\in\mathcal W}\nabla h(x_0)(Ax_0+Bu_0+w) =-r(x_0,u_0)>0$.
By continuity of $x\mapsto\min_{w\in\mathcal W}\nabla h(x)(Ax+Bu_0+w)$, there exist a neighborhood $\Omega$ of $x_0$ and $c>0$ such that
    $\nabla h(x)v\geq c$, for all $x\in\Omega$ and $v\in\mathcal F_{u_0}(x)$.
Hence, along any absolutely continuous solution of $\dot x\in\mathcal F_{u_0}(x)$ that remains in $\Omega$, $\dot h(x(t))\geq c$ almost everywhere. Therefore, $h(x(t))\geq h(x_0)+ct\geq0$ for all sufficiently small $t>0$.

In the idle mode, the LIF dynamics preserve nonnegativity, and continuity keeps both LIF states below their firing thresholds for sufficiently small time. In the application-pending mode, the LIF states remain constant and $\tau(t)=\tau_0-t>0$ for sufficiently small time. Hence, condition (VC) of \cite[Prop.~2.10]{SCC.Goebel.Sanfelice.ea2012} holds at $z_0$, and therefore a nontrivial solution $\phi$ of $\mathcal H$ exists with $\phi(0,0)=z_0$.
\end{proof}
\begin{remark}
The jump maps also prevent consecutive jumps without intervening flow: an initiation jump gives $\tau^+=T>0$, whereas an application jump gives $\xi_s^+=\xi_p^+=0$.
\end{remark}

\subsection{Inter-Event Time Bounds}

Each application jump resets $\xi_s$ and $\xi_p$ to zero. We next bound the time required for either LIF state to reach its firing threshold after an application. Let $\gamma\coloneqq\delta-(a+b)T-\kappa>0$. Since $\psi_p$ is continuous and $\mathcal S\times\mathcal U$ is compact, define
\begin{equation}
    \bar\psi_p\coloneqq\max_{(x,u)\in\mathcal S\times\mathcal U}\psi_p(x,u).
    \label{eq:zoh_drive_bounds}
\end{equation}
Let $\rho_\Delta\coloneqq\Phi_s^{-1}(\Delta_s)\in(0,\kappa]$ and
$t_s\coloneqq(\gamma+\rho_\Delta)/a$. Define
\begin{equation}
    t_p \coloneqq
    \begin{cases}
        \dfrac{\Delta_p}{k_p\bar\psi_p},
        & \mu_p=0,\ \bar\psi_p>0,\\[2ex]
        -\dfrac{1}{\mu_p}\ln\left(1-\dfrac{\mu_p\Delta_p}{k_p\bar\psi_p}\right),
        & \mu_p>0,\ k_p\bar\psi_p>\mu_p\Delta_p,\\[2ex]
        +\infty,
        & \text{otherwise}.
    \end{cases}
    \label{eq:zoh_neuron_dwell}
\end{equation}
The following lemma establishes the event-separation bounds.
\begin{lemma}\label{lem:zoh_miet}
For every solution $\phi$ of $\mathcal H$ with $\phi(0,0)\in\mathcal E$, after every application jump, the next initiation, if it occurs, is separated by at least
\begin{equation}
    t_{\mathrm{idle}}\coloneqq\min\{t_s,t_p\}>0.
    \label{eq:idle_bound}
\end{equation}
Consecutive jumps are separated by at least $\min\{T,t_{\mathrm{idle}}\}$, and consecutive initiation events and consecutive application events are separated by at least $T+t_{\mathrm{idle}}$.
\end{lemma}
\begin{proof}
Let $\phi$ be a solution of $\mathcal H$ with $\phi(0,0)\in\mathcal E$. Fix an application jump at ordinary time $t_0$, and let $j$ denote the post-jump discrete index. By Proposition~\ref{prop:zoh_preinvariance}, $\phi(t,j)\in\mathcal E$ throughout its domain; hence the pre-jump state at the application event belongs to $\mathcal E_a$. By the definition of $\mathcal E_a$, the application jump map, and $v\in\Psi_\epsilon(y)$, $r(x(t_0,j),u(t_0,j))\leq-\delta+bT=-aT-\kappa-\gamma$, while $\xi_s(t_0,j)=\xi_p(t_0,j)=0$.

During the subsequent idle flow, integrating \eqref{eq:residual_drift_a} over $[t_0,t_0+s]$ gives $\psi_s(x(t_0+s,j),u(t_0+s,j))\leq\max\{0,as-\gamma\}$. Since the input is constant on the idle interval, variation of constants yields $\xi_s(t_0+s,j) \leq \Phi_s\!\left(\max\{0,as-\gamma\}\right)$. Since $\Delta_s=\Phi_s(\rho_\Delta)$ and $\Phi_s$ is strictly increasing, $\xi_s(t_0+s,j)<\Delta_s$ for all $s<t_s$. Thus, a safety-driven initiation cannot occur before $t_s$.

For the performance LIF state, $\psi_p(x,u)\leq\bar\psi_p$ on $\mathcal S\times\mathcal U$. Since $\xi_p(t_0,j)=0$, scalar comparison gives
\begin{equation}
    \xi_p(t_0+s,j) \leq
    \begin{cases}
        k_p\bar\psi_p s, & \mu_p=0,\\[1ex]\dfrac{k_p\bar\psi_p}{\mu_p}\left(1-e^{-\mu_p s}\right), & \mu_p>0,
    \end{cases}
    \label{eq:zoh_charge_upper_bound}
\end{equation}
for every $s\geq0$ such that $(t_0+s,j)\in\operatorname{dom}\phi$ and the solution remains in the idle mode. By the definition of $t_p$, $\xi_p(t_0+s,j)<\Delta_p$ for all $s<t_p$. Hence no initiation occurs before $t_{\mathrm{idle}}=\min\{t_s,t_p\}$.

At an initiation jump, $\tau$ is reset to $T$. During the subsequent application-pending flow, $\dot\tau=-1$, so the corresponding application jump occurs exactly $T$ units of ordinary time later. Therefore, consecutive jumps are separated by at least $\min\{T,t_{\mathrm{idle}}\}$, while consecutive initiation events and consecutive application events are separated by at least $T+t_{\mathrm{idle}}$.
\end{proof}

\subsection{Safety, Completeness, and Execution Rate}

The following theorem combines the invariance, continuation, and inter-event-time results established above.

\begin{theorem}\label{thm:zoh}
Suppose Assumptions~\ref{ass:disturbance}--\ref{ass:filter} hold, the bounds \eqref{eq:residual_drift_a}--\eqref{eq:residual_drift_b} hold, and the parameters satisfy \eqref{eq:zoh_margin_condition} and \eqref{eq:zoh_firing_condition}. Then every maximal solution $\phi$ of $\mathcal H$ with $\phi(0,0)\in\mathcal E$ remains in $\mathcal E$ and has unbounded ordinary-time domain. Consequently, $x(t,j)\in\mathcal S$ and $r(x(t,j),u(t,j))\leq0$ for all $(t,j)\in\operatorname{dom}\phi$. Moreover, $\mathcal H$ has no Zeno solutions. If $N_i(t)$ denotes the number of computation initiations in $[0,t]$, then
\begin{equation}\label{eq:boundN}
    N_i(t)\leq1+\left\lfloor\frac{t}{T+t_{\mathrm{idle}}}\right\rfloor.
\end{equation}
\end{theorem}
\begin{proof}
Proposition~\ref{prop:zoh_preinvariance} implies strong forward invariance of $\mathcal E$, and Lemma~\ref{lem:zoh_existence} implies the existence of local solutions starting from every point in $\mathcal E$. By Lemma~\ref{lem:zoh_miet}, every bounded ordinary-time interval contains finitely many jumps.

For the sake of contradiction, suppose a maximal solution has finite terminal ordinary time. Since only finitely many jumps can occur before that time, the solution has a final flow interval. Because $\mathcal E$ is compact and $F$ is locally bounded on $\mathcal C$, $F$ is bounded on $\mathcal E$. Hence, for the last discrete index $j$, the flow arc $\phi(\cdot,j)$ is Lipschitz and has a limit $z_*\in\mathcal E$ at the terminal time. If $z_*\in\mathcal D$, a jump extends the solution. If $z_*\notin\mathcal D$, Lemma~\ref{lem:zoh_existence} gives a positive-duration flow extension. Both contradict maximality. Thus, every maximal solution has unbounded ordinary-time domain.

By Proposition~\ref{prop:zoh_preinvariance}, $\phi(t,j)\in\mathcal E$ for all $(t,j)\in\operatorname{dom}\phi$. On $\mathcal E_i$, $r(x,u)\leq-aT<0$, while on $\mathcal E_a$, $r(x,u)\leq-a\tau\leq0$. Since the $x$-projection of $\mathcal E$ is contained in $\mathcal S$, the safety claims follow. Lemma~\ref{lem:zoh_miet} gives a positive lower bound on the separation between consecutive jumps and therefore excludes Zeno solutions. The same lemma gives the bound on $N_i(t)$ in \eqref{eq:boundN}.
\end{proof}

\begin{remark}
For every $x_0\in\mathcal S$ and $u_0\in\mathcal K_{\rcbf}^{\delta}(x_0)$, the initial state $z_0=(x_0,u_0,x_0,0,0,0,0)$ belongs to $\mathcal E$. Indeed, $r(x_0,u_0)\leq-\delta<-aT-\kappa$ by \eqref{eq:zoh_margin_condition}. Hence $\psi_s(x_0,u_0)=0$ and $\xi_s(0)=0=\Phi_s(0)$, so $z_0\in\mathcal E_i\subset\mathcal E$.
\end{remark}

The next section presents numerical experiments evaluating the safety guarantees, delayed input application, and event-triggered computation rates established above.

\section{Numerical Experiment}
\label{sec:numerical_experiments}

A spacecraft terminal-approach example is used to compare the developed Dual-LIF safety filter with the unfiltered nominal controller (\emph{No-CBF}), a periodic CBF filter (\emph{P-CBF-QP}) \cite{SCC.Ames.Xu.ea2017}, and a delay-tightened periodic CBF filter (\emph{TP-CBF-QP}) \cite{SCC.Singletary.Chen.ea2020}. Under the timing convention in Fig.~\ref{fig:timing}, P-CBF-QP enforces $r(y,u)\leq0$, whereas TP-CBF-QP and Dual-LIF enforce $r(y,u)\leq-\delta$. The periodic filters initiate a QP solve every $T$ seconds, while Dual-LIF initiates a solve when either LIF state reaches its threshold.

The CBF-QPs are solved using \texttt{quadprog} with $\epsilon=0$. Continuous dynamics are integrated using fixed-step fourth-order Runge--Kutta with $\Delta t=T/20$, and LIF thresholds are checked at the integration instants. No-CBF applies the saturated nominal feedback directly.

\subsection{Spacecraft Terminal-Approach Experiment}
\label{subsec:exp_spacecraft_safety}

Consider a spacecraft approaching a non-maneuvering target in a circular reference orbit, modeled using the planar Clohessy--Wiltshire dynamics \cite{SCC.Clohessy.Wiltshire.ea1960}, similar to the relative-motion dynamics used in \cite{BreedenGuaranteed2022}. Let $p_x,p_y$ denote the relative position components and $v_x,v_y$ the corresponding relative velocities. With normalized state $x=[p_x/p_0,\;p_y/p_0,\;v_x/v_0,\;v_y/v_0]^\top$, the dynamics have the form \eqref{eq:plant}, where
\begin{equation}
    A=
    \begin{bmatrix}
        0&0&v_0/p_0&0\\
        0&0&0&v_0/p_0\\
        3\omega^2p_0/v_0&0&0&2\omega\\
        0&0&-2\omega&0
    \end{bmatrix},
    \quad
    B=\frac{a_0}{v_0}
    \begin{bmatrix}
        0_{2\times2}\\
        I_2
    \end{bmatrix}.
\end{equation}
The system parameters are $\omega=1.10\times10^{-3}$~s$^{-1}$, $p_0=10$~m, $v_0=0.2$~m/s, $a_0=0.02$~m/s$^2$, and $\mathcal U=[-1,1]^2$. The disturbance is $d(t)=[0,\;0,\;\bar d\sin t,\;0]^\top$ with $\bar d=0.005$, which satisfies $\|d(t)\|\leq\bar d$. Safety is defined by $h(x)=1-x^\top Px$, with $P=\begin{bmatrix}1.25&0.25\\0.25&0.375\end{bmatrix}\otimes I_2$, so $\mathcal S=\{x:x^\top Px\leq1\}$, where $\otimes$ denotes the Kronecker product. The nominal controller is $u_{\mathrm{des}}(x)=\operatorname{sat}_{\mathcal U}(K_{\mathrm{des}}x)$ with $K_{\mathrm{des}}=2\left[\begin{smallmatrix} -4.00&0&-0.85&-0.022\\ 0&-4.00&0.022&-0.85 \end{smallmatrix}\right]$ and $\alpha(s)=0.5s$. 

The bounds in \eqref{eq:residual_drift_a}--\eqref{eq:residual_drift_b} are evaluated numerically over $\mathcal S$, giving $a=b=0.1338552$. We select $T=0.05$~s, $\kappa=0.001$, and $\delta=1.532\times10^{-2}$, with $\delta_{\mathrm{feas}}\approx1.626\times10^{-2}$ so that \eqref{eq:zoh_margin_condition} holds. The Dual-LIF parameters are $(\mu_s,k_s)=(0.1,100)$, $(\mu_p,k_p)=(5,20)$, $u_{\mathrm{tol}}=0.1$, and $\Delta_p=0.5$. The safety threshold is $\Delta_s=0.001\Phi_s(\kappa)=3.734\times10^{-7}$, which satisfies \eqref{eq:zoh_firing_condition}. With $\Delta t=0.0025$~s, $a\Delta t=3.346\times10^{-4}<\kappa-\Phi_s^{-1}(\Delta_s)=9.684\times10^{-4}$, so threshold detection at the integration instants preserves the required initiation bound. The simulation starts near the boundary of $\mathcal{S}$ at $x(0)=[0.63245,\;0.63245,\;0,\;0]^\top$, for which $h(x(0))=10^{-5}$, and runs for $120$~s.

\begin{figure}[htbp!]
    \centering
    \begin{tikzpicture}

\def\datadir{data}

\begin{axis}[
     width=\columnwidth,
    height=0.5\columnwidth,
    xmin=0,
    grid=major,
    grid style={gray!20},
    axis line style={black},
    clip=false,
    clip marker paths=false,
    ylabel={$\|p\|/p_0$},
    xlabel={$t/t_f$},
    legend pos=north east,
    legend style={nodes={scale=0.75,transform shape}},
]

\addplot[blue] table[x=tau, y=pn_nocbf, col sep=tab] {\datadir/norms.dat};
\addlegendentry{No-CBF}

\addplot[red] table[x=tau, y=pn_pcbfqp, col sep=tab] {\datadir/norms.dat};
\addlegendentry{P-CBF-QP}

\addplot[green!50!black] table[x=tau, y=pn_tpcbfqp, col sep=tab] {\datadir/norms.dat};
\addlegendentry{TP-CBF-QP}

\addplot[violet] table[x=tau, y=pn_duallif, col sep=tab] {\datadir/norms.dat};
\addlegendentry{Dual-LIF}

\addplot[
    only marks,
    mark=diamond*,
    mark size=3.0pt,
    black,
    mark options={fill=black},
    restrict expr to domain={\coordindex}{0:0}
] table[x=tau, y=pn_duallif, col sep=tab] {\datadir/norms.dat};
\addlegendentry{Initial state}

\addplot[
    only marks,
    mark=o,
    mark size=2.5pt,
    black,
    mark options={fill=white}
] coordinates {(1,0)};
\addlegendentry{Target}

\end{axis}
\end{tikzpicture}
    \vspace{-1em}
    \caption{Trajectory of the normalized distance to the target, $\|p(t)\|/p_0$, for the four controllers.}
    \label{fig:exp3_trajectory}
\end{figure}
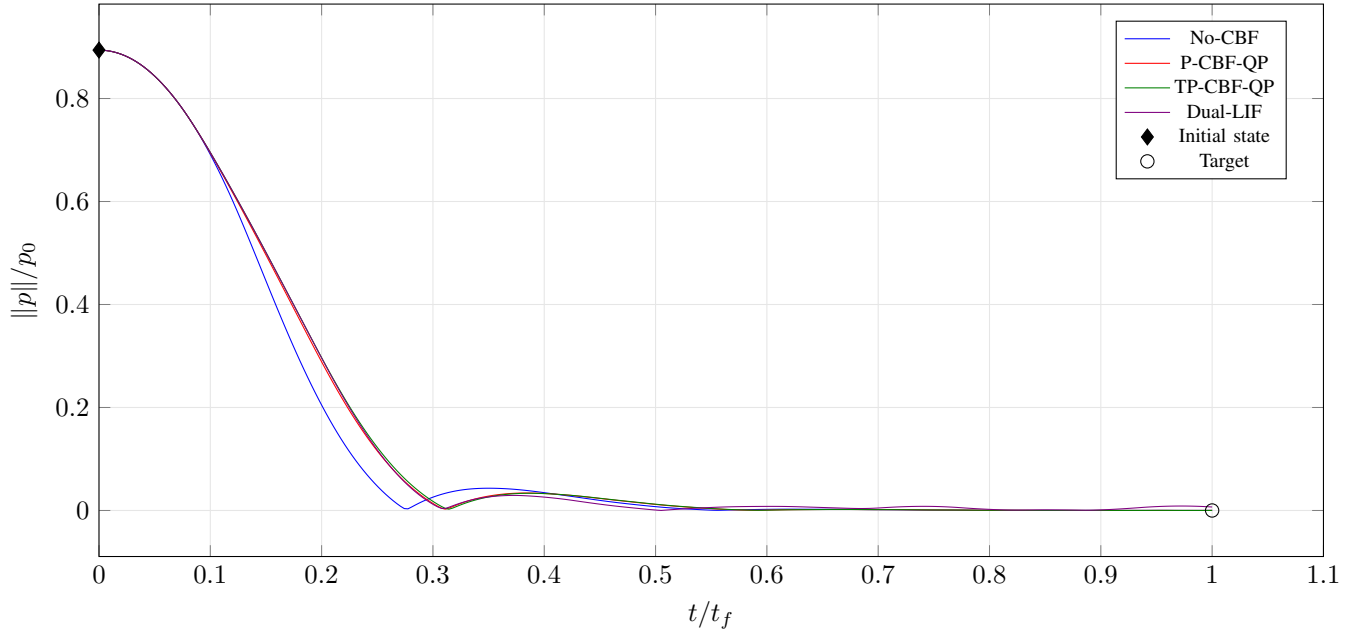


\begin{figure}[htbp!]
    \centering
    \begin{tikzpicture}

\begin{axis}[
    width=\columnwidth,
    height=0.5\columnwidth,
    xlabel={$t$ (s)},
    ylabel={$r(x,u)$},
    xmin=0,
    xmax=120,
    grid=major,
    grid style={gray!20},
    scaled y ticks=false,
    yticklabel style={
        /pgf/number format/fixed,
        /pgf/number format/precision=3
    },
    legend pos=north east,
    legend style={nodes={scale=0.75,transform shape}},
]


\addplot[blue]
table[
    x=t,
    y=r,
    col sep=tab
]{data/traj_nocbf.dat};

\addplot[red]
table[
    x=t,
    y=r,
    col sep=tab
]{data/traj_pcbf.dat};

\addplot[matlabyellow]
table[
    x=t,
    y=r,
    col sep=tab
]{data/traj_tpcbf.dat};

\addplot[matlabpurple]
table[
    x=t,
    y=r,
    col sep=tab
]{data/traj_lif.dat};


\addplot[
    black,
    dashed,
    thin,
    forget plot
]
coordinates {
    (0,0)
    (120,0)
};


\legend{
    No-CBF,
    P-CBF-QP,
    TP-CBF-QP,
    Dual-LIF
}

\end{axis}

\end{tikzpicture}
    \vspace{-1em}
    \caption{Trajectory of the disturbance-adjusted residual $r(x(t),u(t))$. The dashed line denotes the CBF boundary $r=0$.}
    \label{fig:exp3_residual}
\end{figure}
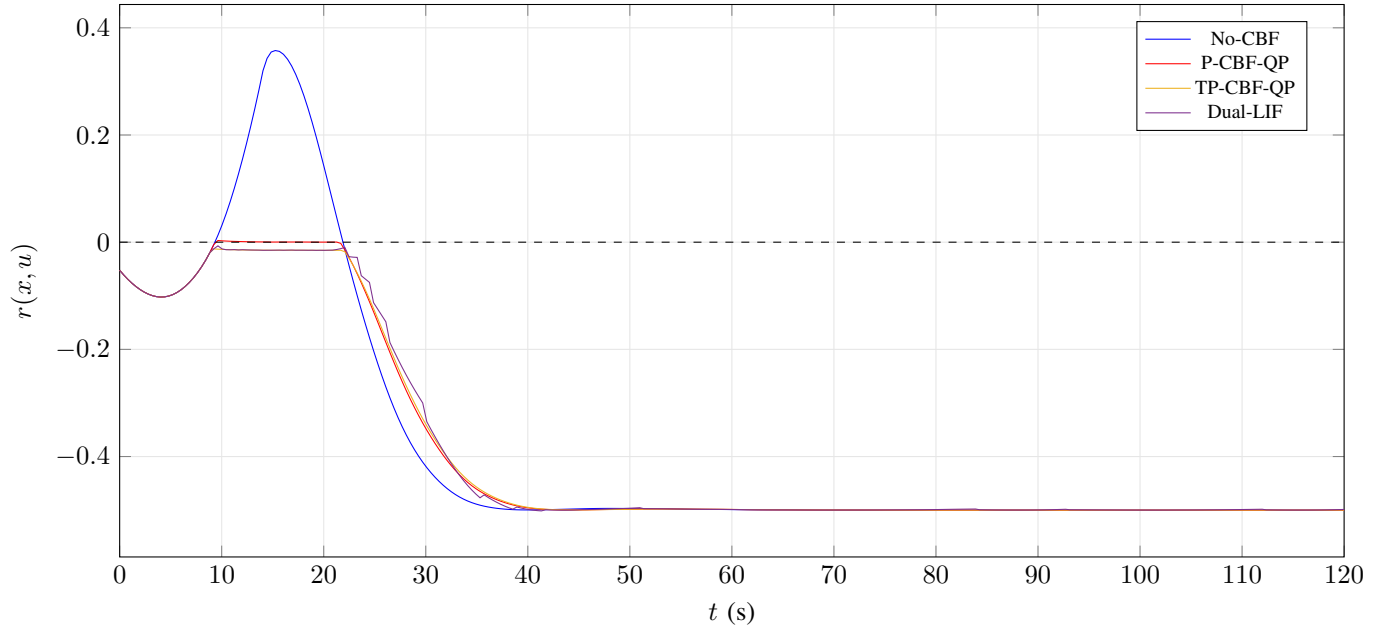

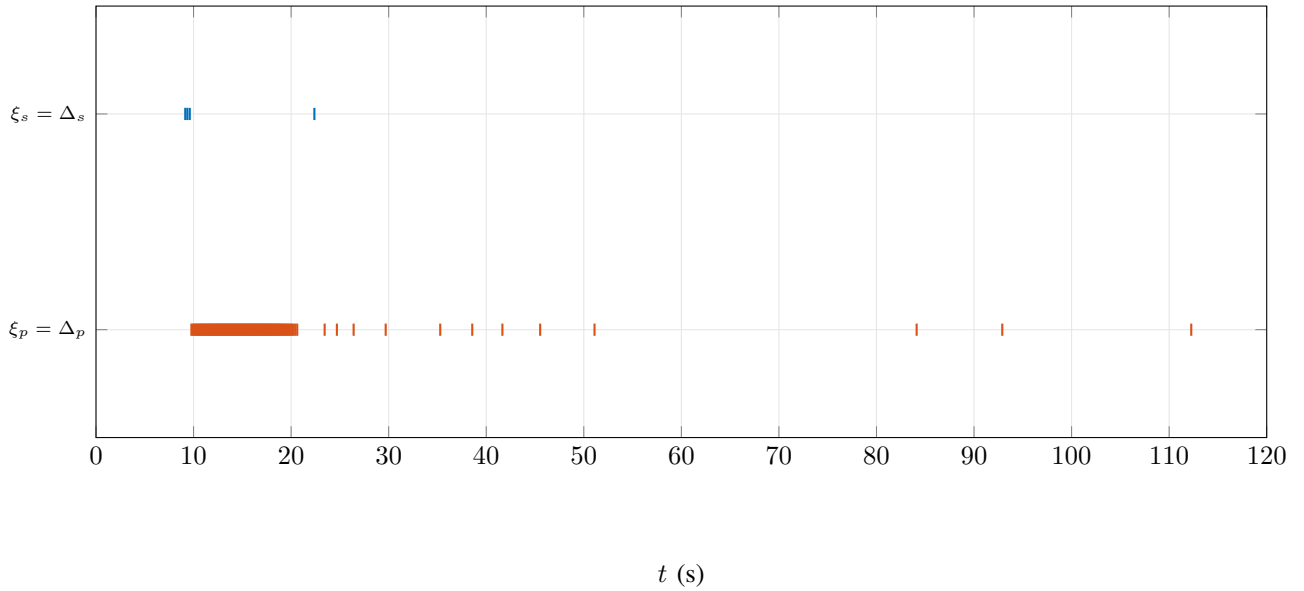
\begin{figure}[htbp!]
    \centering
    \begin{tikzpicture}

\begin{axis}[
    width=0.96\columnwidth,
    height=0.41\columnwidth,
    xlabel={$t$ (s)},
    xlabel style={
        at={(axis description cs:0.5,-0.27)},
        anchor=north
    },
    xmin=0,
    xmax=120,
    ymin=0.5,
    ymax=2.5,
    ytick={1,2},
    yticklabels={$\xi_p=\Delta_p$,$\xi_s=\Delta_s$},
    grid=major,
    grid style={gray!20},
    yticklabel style={font=\scriptsize}
]

\addplot[
    matlaborange,
    only marks,
    mark=|,
    mark size=2.3pt,
    line width=0.8pt,
    restrict expr to domain={\thisrow{type}}{0:0}
]
table[
    x=t,
    y expr=1,
    col sep=tab
]{data/lif_triggers.dat};

\addplot[
    matlabblue,
    only marks,
    mark=|,
    mark size=2.3pt,
    line width=0.8pt,
    restrict expr to domain={\thisrow{type}}{1:1}
]
table[
    x=t,
    y expr=2,
    col sep=tab
]{data/lif_triggers.dat};

\end{axis}

\end{tikzpicture}
    \vspace{-1em}
    \caption{Dual-LIF computation initiation times.}
    \label{fig:exp3_raster}
\end{figure}

\begin{table}[t]
    \centering
    \caption{Safety metrics and online QP solve counts.}
    \label{tab:num_summary}
    \setlength{\tabcolsep}{4pt}
    \renewcommand{\arraystretch}{1.05}
    \begin{tabular}{@{}lrrr@{}}
        \toprule
        Controller & $\min h$ & $\max r$ [s$^{-1}$] & QP Solves \\
        \midrule
        No-CBF    & $-6.060\times10^{-1}$ & $3.583\times10^{-1}$  & 0 \\
        P-CBF-QP  & $1.000\times10^{-5}$  & $4.110\times10^{-3}$  & 2401 \\
        TP-CBF-QP & $1.000\times10^{-5}$  & $-1.160\times10^{-2}$ & 2401 \\
        \rowcolor{lifshade}
        Dual-LIF  & $1.000\times10^{-5}$  & $-5.617\times10^{-3}$ & 142 \\
        \bottomrule
    \end{tabular}
\end{table}

\section{Discussion}

Figure~\ref{fig:exp3_trajectory} shows that all four controllers approach the target. Figure~\ref{fig:exp3_residual} shows different residual behavior across the filtered implementations. P-CBF-QP attains a positive residual, with $\max r=4.124\times10^{-3}$, indicating a temporary violation of the CBF condition under computation delay. In contrast, TP-CBF-QP and Dual-LIF maintain $r(x,u)<0$ throughout the simulation, consistent with the additional margin imposed by the tightened constraint.

Figure~\ref{fig:exp3_raster} shows that $138$ of the $142$ Dual-LIF computations are performance-triggered and $4$ are safety-triggered, so performance updates dominate, while the safety trigger is activated only four times. As summarized in Table~\ref{tab:num_summary}, Dual-LIF reduces the number of QP solution requests from $2401$ to $142$, corresponding to a $94.1\%$ reduction relative to TP-CBF-QP, while maintaining $h(x)\geq0$ and a negative residual. This reduction is specific to the present trajectory, disturbance realization, and trigger tuning.

\section{Conclusion}\label{sec:conclusion}

A dual-LIF architecture is developed for CBF safety filtering with nonzero computation time. Coupling the safety-LIF threshold to residual-growth bounds ensures that computation is initiated before the available residual margin is exhausted. Under the stated assumptions, the hybrid closed loop guarantees robust forward invariance, completeness, and non-Zeno execution.
The spacecraft experiment demonstrates the resulting safety behavior and a $94.1\%$ reduction in online QP solves relative to the delay-tightened periodic implementation.

The analysis assumes exact state information, a known computation-time bound, and sufficient control authority to satisfy the required tightening. Future work will consider state-estimation error, time-varying computation latency, broader operating conditions, and hardware implementation.

\bibliographystyle{IEEETrans.bst}
\bibliography{scc, sccmaster, scctemp, acc_computeconstrained}

\appendix
\section{Technical Proofs}
\label{app:proofs}

\subsection{Proof of Proposition~\ref{prop:filter_well_posed}}
\label{app:prop_filter_proof}

\begin{proof}
Set $\eta\coloneqq\delta_{\mathrm{feas}}-\delta>0$. By the definition of $\delta_{\mathrm{feas}}$, for every $x\in\mathcal S$ there exists $\varpi\in\mathcal U$ such that $r(x,\varpi)\leq-\delta-\eta$. Hence $\mathcal K_{\rcbf}^{\delta}(x)$ is nonempty. Compactness and convexity follow from compactness and convexity of $\mathcal U$, continuity of $r$, and affinity of $r(x,\cdot)$. The graph of $\mathcal K_{\rcbf}^{\delta}$ is closed; since its values are contained in the compact set $\mathcal U$, $\mathcal K_{\rcbf}^{\delta}$ is outer semicontinuous on $\mathcal S$.

To establish inner semicontinuity, let $x_j\to x$ in $\mathcal S$ and $v\in\mathcal K_{\rcbf}^{\delta}(x)$. Define $e_j\coloneqq\max\{0,r(x_j,v)+\delta\}$ and $\theta_j\coloneqq e_j/(\eta+e_j)$. Continuity of $r$ gives $e_j\to0$. Select $\vartheta_j \in\mathcal U$ such that $r(x_j,\vartheta_j)\leq-\delta-\eta$, and set $v_j\coloneqq(1-\theta_j)v+\theta_j\vartheta_j$. By affinity of $r(x_j,\cdot)$,
$r(x_j,v_j)\leq(1-\theta_j)(-\delta+e_j)+\theta_j(-\delta-\eta)=-\delta$.
Thus $v_j\in\mathcal K_{\rcbf}^{\delta}(x_j)$. Since $\theta_j\to0$ and $\mathcal U$ is bounded, $v_j\to v$. Hence, $\mathcal K_{\rcbf}^{\delta}$ is inner semicontinuous on $\mathcal S$.

By Berge's maximum theorem \cite[Section~5.3]{SCC.Bank.Guddat.ea1982}, the inner and outer semicontinuity of $\mathcal K_{\rcbf}^{\delta}$, together with its nonempty compact values and the continuity of $\varphi$, implies continuity of $V^*$ and upper semicontinuity of the optimizer correspondence $x\mapsto\operatorname*{arg\,min}_{v\in\mathcal K_{\rcbf}^{\delta}(x)}\varphi(x,v)$. Since $\varphi(x,\cdot)$ is strictly convex and $\mathcal K_{\rcbf}^{\delta}(x)$ is convex, the optimizer is unique for every $x\in\mathcal S$. Hence the optimizer correspondence is the singleton $\{\pi_\delta(x)\}$, and $\pi_\delta$ is continuous.

Moreover, $\Psi_\epsilon(x)$ is nonempty because $\pi_\delta(x)\in\Psi_\epsilon(x)$. Continuity of $\varphi$ and $V^*$ and closedness of the graph of $\mathcal K_{\rcbf}^{\delta}$ imply that $\Psi_\epsilon$ has a closed graph. Since its values are contained in the compact set $\mathcal U$, $\Psi_\epsilon$ has compact values and is outer semicontinuous on $\mathcal S$.

Finally, optimality of $\pi_\delta(x)$ over the convex set $\mathcal K_{\rcbf}^{\delta}(x)$ gives $(\pi_\delta(x)-u_{\mathrm{des}}(x))^\top(v-\pi_\delta(x))\geq0$. Therefore, $\varphi(x,v)-V^*(x)\geq\frac{1}{2}\|v-\pi_\delta(x)\|^2$. For $v\in\Psi_\epsilon(x)$, $\varphi(x,v)-V^*(x)\leq\epsilon$, and hence $\|v-\pi_\delta(x)\|\leq\sqrt{2\epsilon}$.
\end{proof}

\subsection{Proof of Lemma~\ref{lem:safety_filter_trigger}}
\label{app:lemma_filter_proof}

\begin{proof}
Integrating \eqref{eq:residual_drift_a} over $[t_a,t_b]$ yields $\kappa=r(x(t_b),u)-r(x(t_a),u)\leq a(t_b-t_a)$, so $t_b-t_a\geq\kappa/a$. Moreover, for every $s\in[t_a,t_b]$, integrating \eqref{eq:residual_drift_a} over $[s,t_b]$ yields $r(x(t_b),u)-r(x(s),u)\leq a(t_b-s)$, or equivalently $r(x(s),u)\geq r(x(t_b),u)-a(t_b-s)$, and hence $\psi_s(x(s),u)\geq\max\{0,\kappa-a(t_b-s)\}$.

If $\xi_s(t_a)=\Delta_s$, set $t^*=t_a$. Otherwise, let $t^*$ denote the first time after $t_a$ at which $\xi_s=\Delta_s$, if such a time exists. Suppose no such time exists in $[t_a,t_b]$. Then, by hypothesis, $\xi_s$ evolves according to \eqref{eq:zoh_safety_filter} throughout $[t_a,t_b]$. Variation of constants and $\xi_s(t_a)\geq0$ yield
$\xi_s(t_b)\geq k_s\int_{t_b-\kappa/a}^{t_b} e^{-\mu_s(t_b-s)} \bigl(\kappa-a(t_b-s)\bigr)\,ds =\Phi_s(\kappa)\geq\Delta_s$, where the last inequality follows from \eqref{eq:zoh_firing_condition}. Since $\xi_s(t_a)<\Delta_s$ and $\xi_s$ is continuous on $[t_a,t_b]$, there exists $t^*\in(t_a,t_b]$ such that $\xi_s(t^*)=\Delta_s$, contradicting the supposition. Hence, the safety LIF state fires no later than $t_b$.
\end{proof}

\end{document}